\documentclass[doc,12pt,floatsintext]{apa7}

\usepackage[american]{babel}
\usepackage{csquotes}

\usepackage{amsmath,amssymb,amsfonts}
\usepackage{amsthm}
\usepackage{mathrsfs}
\usepackage{mathtools}

\usepackage{graphicx}
\usepackage{booktabs}

\makeatletter
\@ifundefined{bigtimes}{%
  \newcommand*{\bigtimes}{\mathbin{\vcenter{\hbox{\scalebox{1.5}{$\times$}}}}}%
}{}
\makeatother

\usepackage[natbibapa]{apacite}

\theoremstyle{plain}
\newtheorem{theorem}{Theorem}
\newtheorem{proposition}[theorem]{Proposition}

\theoremstyle{definition}

\theoremstyle{remark}

\title{Causal Inference for Case Studies in Behavioral Health}
\shorttitle{CAUSAL INFERENCE FOR CASE STUDIES}

\author{Shane W. Sparkes}
\affiliation{Wolf \& Noble, Inc.}

\authornote{%
Shane W. Sparkes, Wolf \& Noble, Inc., Huntsville, Alabama, United States.

The author declares that he has no conflict of interest.

Correspondence concerning this article should be addressed to Shane W. Sparkes, Wolf \& Noble, Inc., Huntsville, AL 35808, United States. Email: ssparkes@wolfnoble.com}

\abstract{We present a framework for causal inference in behavioral health case studies---and observational N=1 settings more generally---under unmeasured confounding. The framework rests on a class of causal estimands, termed $\Omega$ estimands, defined as contrasts of functions of an outcome variable's support rather than its distribution. Because such estimands do not depend on how probability is distributed over supports, they are insensitive to the confounding that limits other methods. We prove that, in a structural causal model, the observational and interventional supports of an outcome coincide under a single assumption---positivity---without any requirement that confounders be known, measured, or adjusted for. Two optional conditions extend the framework: one licensing a client's recalled baseline as a stand-in for sparsely measured baseline periods, and one connecting support contrasts to conventional mean contrasts through an expected-average identity. We adopt a subjectivist (de Finetti) interpretation of probability and situate the framework within mandates for measurement-based care. A case study of cognitive behavioral therapy for anxiety illustrates an elementary approach a provider can use.}

\keywords{causal inference, case studies, N=1, positivity, outcomes analysis, behavioral health}

\begin{document}
\maketitle

\section{Introduction}
This manuscript provides a route to causal inference for observational $N=1$ studies, i.e., case studies. More generally, it provides a route to identifying causal estimands in any setting under one condition only: positivity, defined and proven sufficient in Section~3. This is advantageous, as usually much heavier conditions are required.

Individual care processes in behavioral health motivate this paper. Put succinctly, an individual care process in behavioral health is one such that a set of providers enact a set of clinical interventions to improve client outcomes over time. When care is measurement-based, it also produces a sequence of measurements that warrant outcomes analysis. Such an analysis requires a different set of methodologies than population-based studies because a care process lacks randomization, experimental control, and a population frame. The methods developed here provide an avenue for cogently assessing treatment efficacy within an individual process despite the absence of these components.

Since our applied domain is behavioral health practice, mostly discrete outcome variables are considered. To this end, we define an indexing set $I=\{1, 2, 3, \ldots, m \}$ for some $m \in \mathbb{N}$ that represents potential observation points in an intervention process. This set indexes the stochastic collection $\mathcal{O} = (\mathbf{Y}, \mathbf{T}, \mathbf{X}, \mathbf{U})$, where $\mathbf{Y}$ is usually a $m \times 1$ vector of integer-valued ordinal outcome variables, $\mathbf{T}$ is a $m \times 1$ vector of intervention variables, $\mathbf{X}$ is a $m \times p$ matrix of additional features, and $\mathbf{U}$ is a $m \times q$ matrix of unmeasured variables. It is not assumed that the $Y_i$ or $T_i$ for $i \in I$ are identically distributed.

A few basic definitions are required before introducing the estimands of interest. For an arbitrary $i$, define $Y_i| \text{do}(t)$ as the outcome variable $Y_i$ under the condition that $T_i$ has been experimentally set to $t$, or that it has been intentionally set to $t$ in a manner that accords with a design and that is invariant to background factors. In the language of a structural causal model (SCM), we then say that $Y_i | \text{do}(t) \sim \text{Pr}\{ Y_i = y | \text{do}(t) \}$. This object is often juxtaposed to a version of the outcome variable that has been naturally observed without the manipulation described: $Y_i | t \sim \text{Pr}\{ Y_i = y | t \}$. Literature on the SCM and its basic machinery is available elsewhere \citep{pearl2010foundations, greenland1999causal, pearl2009causality, pearl2010introduction, hernan2010causal}.

Next, we define supports. Recall that the \textit{support} of a random variable $Y_i$, henceforth denoted as $\text{supp}(Y_i)$, is the smallest closed set such that $\text{Pr}\{ Y_i \in \text{supp}(Y_i) \} = 1$. When $Y_i$ is discrete, it is the closed set of $y$ values such that $\text{Pr}(Y_i=y) > 0$. For any stochastic collection $\mathcal{O}$, we can define a joint support for the outcome variables: $\text{supp}(\mathbf{Y}) \subseteq \bigtimes_{i \in I} \text{supp}(Y_i)$, where $\bigtimes$ represents the Cartesian product. The type of support that interests us more is the ``flattened'' one that encompasses all values that are possible as measurements over time with respect to $\zeta = \{Y_1, Y_2, Y_3, \ldots, Y_m \}$. If an arbitrary integer $y$ appears in any of the $m$-tuples of $\text{supp}(\mathbf{Y})$, it is observable with positive probability with respect to $\zeta$, and we collect such values into a set $\text{supp}(\zeta)$. For simplicity, and because we mostly reference discrete outcome spaces, we use $\Omega$ for support sets: $\Omega_{\zeta}$ for the support of a sample and $\Omega_{Y}$ for the support of a random variable when context requires.

Now we can consider our estimands, which we call $\Omega$ estimands. Let $g$ be any function of interest. Then $g(\Omega)$ is a function of the support that does not incorporate probabilities, at least insofar as they are non-uniform. For instance, when $g(\cdot) = \text{max}(\cdot)$ and $\Omega = \{1, 2, 3, 4, 5 \}$, it is the case that $g(\Omega) = 5$. Other meaningful examples include $\text{min}(\cdot)$, the cardinality $|\cdot|$, $\text{range}(\cdot) = \text{max}(\cdot) - \text{min}(\cdot)$, and $\text{median}(\cdot)$, the empirical function applied to an ordered set of constants. Functional averages are another example: the simple arithmetic average of values in a support. When we `take the functional average of $Y$,' symbolized by $\text{Av}(Y)$, we compute the following, where $|\cdot|$ signifies cardinality when used with a set:
$$\text{Av}(Y) = |\Omega|^{-1} \sum_{y \in \Omega} y.$$
Note that $\text{Av}(Y)$ and $\text{Av}(\Omega)$ represent the same object.

The treatment effects that we are interested in are contrasts of $\Omega$ estimands, which we label $\Omega$ effects. For any function of interest $h$, outcome variable $Z$, and contrasting treatment values $t, t'$, an $\Omega$ effect is defined as $h(\Omega_{Z|\text{do}(t)}, \Omega_{Z|\text{do}(t')})$. They can also be defined with respect to flattened supports, i.e., as $h(\Omega_{\zeta|\text{do}(t)}, \Omega_{\zeta|\text{do}(t')})$. In this paper, we often default to a binary $T$, and hence to $h(\Omega_{\zeta|\text{do}(t=1)}, \Omega_{\zeta|\text{do}(t=0)})$.

Changes in the support of an outcome variable following intervention are commonplace. An efficacious blood pressure medication will prevent the materialization of morbidly high measurements for some people; anxiolytic medications can temporarily eliminate extreme levels of anxiety. Most importantly, an effective course of behavioral health treatment can alter consumer experience: a fact that arguably implies the elimination of certain qualia from consumer phenomenology, and hence also their inclusion in outcome spaces.

The remainder of this paper is organized as follows. Section~2 provides background on the applied context and on the challenges---chiefly confounding and matching---that limit cogent causal inference in this domain, and introduces the subjectivist perspective employed for interpreting all conditions and results. Section~3 proves our main proposition and discusses what positivity and the other conditions mean in a behavioral health setting. Section~4 demonstrates the results with data from a clinical case study. Section~5 discusses scope, limitations, and future directions.

\section{Background and Challenges}

The Substance Abuse and Mental Health Services Administration (SAMHSA) is encouraging the adoption of whole-person, measurement-based behavioral health care in the United States, embodied in the archetype of the Certified Community Behavioral Health Clinic (CCBHC) \citep{ccbhc2022, brown2021implementation, mauri2024characterizing, hu2021brief}. Since behavioral health challenges can possess multiple interacting etiologies \citep{uher2017etiology, susser2006psychiatric, borsboom2017network} and because each client materializes a whole life-process of accrued causes and semantic frameworks, CCBHCs provide a breadth of flexible services---case management, primary care coordination, peer support, psychotherapy, and access to psychiatric services---adapted to each person. CCBHCs also possess a mandate to continuously monitor program quality and efficacy through quantitative measures administered at regular intervals. Implementing regular measurement and outcomes analysis is an important step toward a higher quality of care \citep{frank2023toward}.

The analysis of treatment efficacy involves causal inference. However, a felicitous use of an accrued body of measurements for causal reasoning---and ultimately for treatment decision-making---requires a conceptual scaffolding that can withstand the obstacles of this domain. Since every care process is essentially an N-of-1 longitudinal observational study, a firm elementary ground is often elusive, especially with respect to standard statistical methodologies. This section explicates two main challenges: identification and dependence. First, however, it orients the reader to what probability \textit{means} in our setting.

\subsection{A Finettian Perspective}

In this paper, we wholly endorse the famous statement of Bruno de Finetti, here paraphrased: `PROBABILITY DOES NOT EXIST' \citep{deFinetti1974Theory}. Probability obviously exists in some sense because we reference it as an object. What is asserted is something different: that probability does not exist outside of the structures of our experience. Probability is a personal measure of uncertainty. This matters, first, because it changes how we interpret the conditions of Section~3. The proofs remain valid under other interpretations of probability, but the meaning of their statements changes radically. Second, it is important because this philosophy of probability centers data in a way that is congruent with intuition.

A measurement-based care process produces a collection of data, $\mathcal{D}$. This collection is not a \textit{sample} in the traditional sense. Providers do not define a sampling frame and then choose moments to observe and intervene in accordance with a stochastic mechanism. Further, we do not use superpopulations since they are a metaphysical incursion imported for mathematical convenience. What is undeniable in any setting is that each datum is recorded as a representation of being as experienced. Random variables arise only as a function of our \textit{uncertainty} in the designation.

A clinical example is useful. A therapist draws attention to their client's accomplishment with a friendly tone. To an outsider, this is a nicety; for the provider, it is a clinical intervention that seeks to assist the client to reframe their self-understanding. Let $t$ represent this therapeutic action, which the therapist records. However, on reflection, they are not certain that they succeeded in enacting it. They \textit{did} something; this much is undeniable. They are unsure of the \textit{essence} of what they did. Hence, they register $T$, a random variable: somewhat sure that $T=t$, but not certain. This is the type of random variable a behavioral health provider encounters.

This awareness alleviates the notion that a random variable necessarily collapses into `its observation.' What we actually possess from measurement-based care is $\mathcal{D} = (\mathbf{y}, \mathbf{t}, \mathbf{x})$: a rigid record, an anchor, and necessarily an element of the support of $(\mathbf{Y}, \mathbf{T}, \mathbf{X})$ insofar as it is accepted as record. The recording of $\mathcal{D}$, however, suggests no automatic elimination of uncertainty in its representation of what is, or was, present in experience.

On this point, it is useful to differentiate a random variable's \textit{scale} from its \textit{support}. A scale is a function that maps one form of being into another in accordance with a unit. Per that unit and the system of relations it renders present, it possesses a co-domain of values. The \textit{support} of a random variable is necessarily a subset of the co-domain of the scale, but they are only incidentally equal. Although a typical pain scale goes from 0 to 10, a person in intense pain will select only 8, 9, or 10 with non-zero probability. In the minimum, it is self-given that they would never select zero insofar as they are coherent. This distinction is important because many might interpret $\Omega$ effects as automatically null via the faulty equivocation of scales and supports. We return to this point in Section~5.

\subsection{The Problem of Identification}

The challenge of causal inference is well understood. One wishes to know if $T=t$ necessarily unfolds a certain presence of being, $Y=y$; stated counterfactually, that if $t$ had not become present, then neither would have $y$. In the world of our observation, we only possess $(t, y)$: one cannot empirically measure $\delta = y|\text{do}(t=1) - y|\text{do}(t=0)$ for the same unit \citep{Holland1986}. In a behavioral health setting, we cannot rewind and observe what would have happened had a client not received treatment under the very same conditions. Thus, there cannot exist a strictly empirical science of cause and effect.

Beyond this fundamental challenge, treatment efficacy in a care process faces another. Providers do not actualize strict interventions in an experimental sense. A potential client arrives in relation to a labyrinth of unknown contextual variables. After a risk screening or assessment, a baseline score $y_0$ is recorded previous to the commencement of contracted treatment; each subsequent $y_i$ is observed during encounters actualized in accordance with client and provider choice-in-context. Hence, what we have at hand is implicitly $y_0|(t=0)$ and $y_i|(t=1)$ for all $i \in I$: recordings that result from activities unfolding in dialogue with latent factors. At all points, provider knowledge is incomplete, and every encounter unfolds as a surprise of disposition in relation to background happenings. At no point does there exist control in the typical epistemic sense.

Two challenges thus present themselves. One, these values might differ from what would have been observed under strict intervention, i.e., under $\text{do}(t)$. This is important because strict intervention is what suffices for judging that the relationship between treatment and outcome is uninfluenced by background factors. Two, we cannot possess an empirical measurement of the form $y_i|(t=0)$ when we have $y_i|(t=1)$ in hand for a given $i$. Further, a baseline score necessarily manifests at a time and in a space separate from an arbitrary $y_i$. We can contrast their properties to evidence a causal effect only under additional constraints that are extrinsic to the records themselves.

SCMs address the first challenge. Under the condition that $\mathbf{X}$ is sufficient for closing all backdoor paths in the directed acyclic graph (DAG) of the SCM, it follows that $y|t, \mathbf{x} = y|\text{do}(t), \mathbf{x}$. Stated otherwise, by conditioning on a set of variables that blocks the influence of other variables on both $T$ and $Y$, the resulting DAG becomes exchangeable with the DAG under $\text{do}(T=t)$, given $\mathbf{X}=\mathbf{x}$. This conditional equation breaks down when $\mathbf{X}$ is not sufficient for closing all backdoor paths. When this is the case, there still exists a vector of unmeasured variables $\mathbf{U}$ that impacts both treatment and outcome, as in our described process of care.

The second problem is often addressed via matching conditions. Population-based studies exist under the enforced assumption that the observed outcomes $\{ y_i\}_{i \in I}$ are separate observations of the same essential class of being with respect to the properties of interest, regardless of the time and space of measure. This condition is the silent necessity of population formation. It suffices for `observing the same (type of) entity under contrasting conditions,' so that contrasting $\text{E}\{Y_i | \text{do}(t=1) \}$ and $\text{E}\{Y_j | \text{do}(t=0) \}$ stands in for contrasting $\text{E}\{Y_i | \text{do}(t=1) \}$ and $\text{E}\{Y_i | \text{do}(t=0) \}$ for an arbitrary $i \in I$.

Constraints of this character are often intuitive and reasonable. Nevertheless, they are a non-empirical imposition of an often incidental ontology. We can do without them, but the result is a much weakened state of causal inference. Even if all intervention objects are identified, contrasting them does not stand in for a direct counterfactual comparison, since `\textit{ceteris paribus}' is not substantiated. Arguably, the contrast often \textit{does} still yield evidence of a causal relation, but it is weak in an informal inductive sense. Imagine two comparable studies on the relief of headache pain by aspirin: one when we are eighteen (assigned to control) and another when we are fifty (assigned to the medication). Much of our pain experience has changed between them. Taking the contrast at face value, we can assert a difference between the scores of \textit{that} person at \textit{that} time and \textit{this} person \textit{here}, and, crucially, that these two individuals are \textit{similar}. This much is self-given, although it does not establish exactly what we want. Asserting that such a contrast offers no experimental evidence at all is not self-given, however. Indeed, this quality of similarity is essentially what population-based studies achieve when they match individuals on demographics or estimated propensity scores. At best they approximate an identity. Essential identity is only ever \textit{formally} enforced; it can never be empirically enforced. The strength of the inference from similarity is context-dependent.

\paragraph{The obfuscation of the clinic} To identify causal effects per the mechanics of a SCM, we can measure a vector $\mathbf{X}$ sufficient for closing all backdoor paths. To establish a `\textit{ceteris paribus}' counterfactual contrast, we require constraints that match records or random variables. Only the second is feasible in our domain, and even then, under duress.

The confounding of a measurement-based care process becomes further evidenced when one considers that an ethical and therefore felicitous care process requires a \textit{minimization} of formal measurement; and that this dictate coexists with the fact that behavioral health treatment is a function of client-provider \textit{relationships}. A given social relationship---including a professional one---is necessarily framed by, and manifest in accordance with, a web of other relations and variables. A non-trivial subset of this web \textit{possibly} influences both treatment and outcome, and \textit{this} is sufficient for confounding.

As a clinician, I do not continuously measure how societal forces impact me, my service delivery, and the outcomes of the client; nor do I conscientiously or exhaustively measure how smaller, more local variables impact my clinical actions and client experience. This is not a flaw in my practice. If an arbitrary provider decided to stop and formally measure all relevant assemblages of variables that \textit{might} be impacting treatment and outcome, such a practice would destroy the care process with probability one. A set of possible confounders exists in every moment, and in every moment's horizon. Thus, in a care process, we can safely assume that there exists some vector of unmeasured confounders, $\mathbf{U}$.

For a given individual care process, matching is also precluded in any straightforward way. An arbitrary client arrives to a provider untreated by that provider, yielding baseline measures $Y_0$; treatment commences after assessment. $Y_0$ and any measure juxtaposed with it are hence separated by time and space. In the absence of additional intellectual manipulations or constraints, `\textit{ceteris paribus}' can be present only as an enforced contingency. This does not rule out cogent causal inference, but it weakens it.

\paragraph{Related work}

N-of-1 studies possess an expanding literature, especially given growing interest in personalized care \citep{backman1999case, kwasnicka2019challenges, fountzilas2022clinical, kazdin2020single, lundervold2000best, kratochwill2013single, kratochwill2015single, duan2013single, duan2022personalized, kravitz2022conduct}. This literature almost invariably relies on a blocking design that divides a care process into phases of treatment: $A$ for a period without treatment, $B$ for one type of treatment, and so forth. The typical care process is an `AB' study in this language. Designs that aspire to more cogent scientific knowledge repeat each phase and often randomize the sequence (e.g., $ABBAAAB$). Then, under the assumption that observations within the same phase label constitute the same population, or that a stable `parallel trend' is present, the researcher contrasts observations or their summary functions, with data visuals often utilized as an intuitive basis for inference \citep{kazdin2020single}.

This literature has much to impart to the study of intra-individual change, which is critical to understanding causality as it presents itself in radically unique processes \citep{molenaar2004manifesto, molenaar2009new, molenaar2015relation}. However, the design concept is mostly exogenous to our domain, which is the natural individual care process that unfolds for its own sake, and not for the sake of knowledge-construction. Randomization and the intentional withholding of treatment are either unethical or subversive of what is essential to the process. Moreover, the construction supposes that treatment effects `wash out' so that a new state can be observed anew. Yet, almost all treatment approaches presuppose that they impart \textit{lasting} dispositional changes. This approach to N=1 design would hence become applicable under a null hypothesis only, whereas a care process demands \textit{positive} evidence of treatment efficacy, not merely a rejection of a null hypothesis.

Note also that this structure of observation does not alleviate the matching problem; it multiplies it. There is no such thing, \textit{in and of itself}, as a repeated observation of a phase. A client who discontinues treatment and recommences later yields, say, $A_0 B_0 A_1 B_1$, but the person who presents at $A_1$ does so conditional on $A_0$, $B_0$, and other unknown events. They are self-given as similar to the person at $A_0$ by function of what is stable of their body of identity, but they are not the same.

\subsection{The Problem of Dependence}

A final quandary is probabilistic dependence. The N-of-1 literature identifies dependence as a challenge since all data points originate from the same person. When combined with the fact that a typical care process possesses relatively few measurements, this dependence limits the attenuation of uncertainty with respect to summary statistics. As a result, more heuristic methods of analysis are often employed \citep{kazdin2020single, de2026n}. We concur, and we note that the problem runs deeper than is usually acknowledged. Once probability is accepted as a measure of personal belief, the modeling agent is a shared cause of every random variable they author. Since certain self-knowledge is patently absurd, uncertainty about one's own epistemic state induces an association among \textit{all} of one's random variables by construction. Accounting for 'author effects' is therefore an imperative. An authentic DAG that includes the modeler as a common cause would therefore exhibit ubiquitous confounding in the absence of conditional self-measure. We do not develop this argument further here since it is our goal to alleviate this burden. It suffices for our purposes to recognize that unmeasured confounding is inescapable in the clinic for the practical reasons given above. What follows is the need for a route to causal inference that does not require closing backdoor paths at all.

\section{A Limited Solution}

Confounding and matching are the major obstacles to a valid framework of causal inference for individual care processes: the former especially because it is in the provider's and client's interest to minimize measurement and therefore render large sets of confounding variables active. In this section, we prove that a route to causal inference exists that sidesteps this problem because it does not require the measurement of confounders.

First, some notation. Let $\sigma_{X, Y}$ stand in for $\text{E}(XY) - \text{E}(X)\text{E}(Y)$ and $f(\cdot)$ represent a mass function or density. Furthermore, let $R=\int_{x \in \Omega} dx$ or $|\Omega|$ for continuous and discrete outcome sets, respectively. We require one defined entity: $\kappa = R^{-1} \sigma_{X, f(X)^{-1}}$, where $X$ is distributed according to $f(x)$. The targeted covariance relates to sum-symmetry and skew in a probability distribution \citep{sparkes2024functional}. Note that $\text{E}\{ f(X)^{-1} \} = R$ and hence $R^{-1}$ standardizes $X\cdot f(X)^{-1}$. For now, the reader may understand $\kappa$ as a value that captures sum-symmetry and aspects of skew. Further exposition follows below.

Next, we provide the main conditions used in this paper for identifying causal estimands. Of these, we truly require only the first. The remaining two open up more routes for inference.

\begin{enumerate}
	\item[I.] Positivity: Observe $U$ with $\Omega_U$ and $T$ with $\Omega_T$. Then for all $t \in \Omega_T$ and $u \in \Omega_U$, $f(t|u) > 0$.
	\item[II.] Epistemic preservation: Let $t, t' \in \Omega_T$ such that $t \neq t'$. Then $\kappa_{Y|t} = \kappa_{Y|t'}$.
	\item[III.] Support recall: Let $Y$ be a summary outcome of interest. Denote $\hat{Y}$ as an estimate of $Y$. Then $\Omega_{\hat{Y}|t=0} = \Omega_{Y|t=0}$.
\end{enumerate}

Positivity is a common assumption in the potential outcomes framework \citep{hernan2010causal}. It is usually interpreted to mean that all values of $t$ are present in the population strata created by $U=u$. Our understanding does not differ markedly in motif. However, since our philosophy of probability dictates that it is a measure of belief, positivity implies that, whenever we are uncertain of $T$, no $u \in \Omega_U$ \textit{eliminates} values of $t$ from possibility. In other words, $\Omega_{T|u} = \Omega_T$ for all $u \in U$. Equivalently, $f(t, u) > 0$ for all $t, u$, i.e., $\Omega_{T, U} = \Omega_T \bigtimes \Omega_U$. We frame positivity with respect to $u$ only, but it can be extended to $f(t|x, u)$. Since most providers will not consistently measure covariates formally, we default to the expression with just $u$. Fortuitously, one can always theoretically close the backdoor path in a DAG with just $T$, $Y$, and $U$.

An exploration of what this assumption means for practitioners is helpful. Essentially, the condition is fulfilled when the possibility of a treatment, or the possibility of a (vector-valued) confounder, \textit{does not eliminate the possibility of values in its complementary variable}. Consider a concrete failure case. A therapist marginally has a particular set of interventions to choose from when classifying their actions. In one session, however, the client arrives in crisis, and we intuit that background factors---atrocious air quality after a fire, or intoxication---are exerting influence on the client's state in addition to the delivery of services. In the presence of a crisis, the therapist will not engage the client in a casual exercise that maps out thoughts. This is because the background factors that lead to crisis cannot ethically co-exist with positive probability with a casual intervention not attuned to de-escalation. This state of affairs generalizes since certain situations can preclude representations of interventions. From this thought experiment we extract a general rule of thumb: the more granular we become in our representation, the more likely it is that positivity will be violated in some way, at least per our intuition, until we reformulate things cogently.

This point is related to how client-centered care can \textit{limit} our ability to intervene in a strict sense. Our ethics dictate that interventions remain \textit{responsive} to a client's disposition, which is influenced by a constellation of unknown forces. Interventions must therefore remain in dialogue with this same constellation. This is the conundrum of person-centered care, and it conflicts with \textit{evidence-based} care that purports to build a causal understanding for decision-making from process data.

Nevertheless, certain coarse conceptualizations of clinical intervention do fulfill positivity. If the definition of $T$ captures \textit{what is essentially invariant} of intervention as treatment unfolds---what is instituted in practice \textit{by design} and \textit{by convention}, regardless of the client's disposition or identity---then $T$ can adopt any of its values in the presence of any $u \in \Omega_U$. For instance, say the provider is a case manager who provides two types of services: linkages to resources and planning. In accordance with their training, the case manager \textit{always} enacts both in every session. This is true even in the event of a crisis, once they have acquired safe distance. Hence, insofar as we define $T$ as binary---$T=1$ when case management services are provided and $T=0$ when they are not---we know that $0, 1 \in \Omega_T$ for a given encounter, and most certainly when $T$ is abstracted further as a simple marker of exposure over periods of time. Insofar as $U$ \textit{remains valid as a representation}, i.e., the DAG remains felicitous as a map of what is essential of the world with respect to the $(T, Y)$ relationship, positivity is necessarily fulfilled. It would present as violated only if the DAG itself lost credibility, i.e., if the set of relevant contextual factors, as informally intuited, \textit{changed} between periods. In that situation, however, one authors a new map.

A word is due on our employment of uncertainty, for treating $T$ as uncertain might seem strange. What is \textit{essential} of an action is often irreducible to what can be empirically operationalized and perceived. The speech act framework of Austin and Searle is useful here \citep{austin1962, searle1979expression}. A speech act possesses its \textit{locution} (what is said), \textit{illocution} (its intent of action), and \textit{perlocution} (its effect). We generalize this structure to interventions. In a given session, a provider knows \textit{they did something}. However, they often do not know with certainty to what degree they \textit{successfully} actualized their intent. Our case manager hands over a paper with numbers for food banks. Whether this represents \textit{a linkage to resources} is context- and client-dependent. For instance, if the relationship is strained, the client might experience it as a placation. In another example, perhaps the provider might experience their own action as perfunctory. What is essential of the service is the manner in which the \textit{relationship} is utilized to \textit{actualize} the intended ends, and it is this actualization that is susceptible to uncertainty almost ubiquitously. It is also precisely where confounding can obscure things since even a planned intervention is actualized in dialogue with context, which can change our understanding. Herein is where the distribution of probabilities can be altered. Positivity remains fulfilled insofar as the provider does not authentically detect that extraneous forces precluded the very possibility of the illocution or intent generally: such as when, say, a distracting national event reduces a session to peer-like conversation such that no therapy occurred at all. An event like this may violate positivity for that session, however, without violating it for the coarse binary encoding of the exposure over time, insofar as the provider or client authentically determines that \textit{some} treatment was in fact enacted.

A concise takeaway is this: positivity holds when observed treatment is conceptualized at a sufficiently coarse level, intent-effect uncertainty exists per this conceptualization of possibilities, and the DAG remains a valid representation of world(s).

For the condition of epistemic preservation, remember that, as much as we wish to frame our uncertainty of objects \textit{as if} that uncertainty exists `in the past,' all designations of probability distributions manifest necessarily in the present. This follows from the fact that a measure of belief cannot exist in a discorporate state. Thus, although $y_0$ and some function of treatment-exposed outcomes came into existence as records at separate moments, any given expression of $Y_0$ and the treatment-exposed function \textit{as random variables} necessarily arises in the unfurling present. The notion that they may exhibit the same sum-symmetric or skew-like trend, as represented by $\kappa$, is therefore less incredulous than it first appears. The condition is likely to hold if the same agent, in the same greater moment, assigns distributions of uncertainty with the same structural motif. Many individuals assign their probabilities in a manner that is insensitive to the specificity at hand. For instance, they might focus on the value they deem most probable and assign all other probabilities symmetrically around it. The core point, symmetry aside: \textit{if it is reasonable to assert that the agent employs the same shape of distribution relative to its range of possible values, this condition plausibly holds}. It holds \textit{exactly} when the agent assigns probabilities such that the expected value always maintains the same position relative to its support.

We make this statement exact. A useful identity is as follows \citep{sparkes2024functional}:
$$\text{Av}(Y) = \text{E}(Y) + \kappa.$$
We call it the expected-average identity. Thus $\kappa = \text{Av}(Y) - \text{E}(Y)$, and $|\kappa|$ measures the distance between the functional average and the expected value. A preservation of $\kappa$ across treatment and non-treatment indicates that the agent assigns probabilities with the same structural motif.

To aid in the consideration of this condition, we introduce the concept of $\mathbb{U}$ random variables: a class of variables with sum-symmetric probability distributions, explored in detail elsewhere \citep{sparkes2024functional}. When $Y$ is a $\mathbb{U}$ random variable, $\text{Av}(Y) = \text{E}Y$ and therefore $\kappa=0$. All $Y$ with symmetric distributions are in the $\mathbb{U}$ class, though not conversely. Hence, any consumer who selects a value $y$ in the `center' of $\Omega$ with the highest probability and assigns other uncertainties symmetrically around it constructs a $\mathbb{U}$ random variable. So does any consumer who is indifferent and assigns probabilities equally. For supports that are a single interval of real numbers or an unbroken chain of integers, membership in the $\mathbb{U}$ class is equivalent to the areas above and below the cumulative distribution function being equal. If $Y_0$ and its contrasting function of treatment-exposed random variables are both in the $\mathbb{U}$ class, condition II is fulfilled. This will often be the case for empirical means or \textit{estimates} of them. In place of a simple baseline score, the provider can ask the client to estimate the average value of their scores across a relevant period previous to treatment. This reasonably affords a contrast of $\mathbb{U}$ variables.

Condition III provides our bridge to matching. Again let $T=1$ indicate that treatment is active. A provider naturally has access to some summary function $S_{t=1} =h(\mathbf{Y}|t=1)$, but not usually to $S_{t=0} = h(\mathbf{Y}|t=0)$, since the baseline period is typically brief and the provider had no contact with the client before it. One fix is to have the client estimate this measure, i.e., to specify $\hat{S}_{t=0}$ during the baseline period before treatment commences. The bridge arrives via the assumption that the outcome space of $\hat{S}_{t=0}$ equals the outcome space of $S_{t=0}$, i.e., that even though the client estimates the summary of their outcome experience from memory, and this estimate might differ in its distribution of probabilities, it nevertheless draws upon the same possible world of values.

Note that many validated scales already rely on a client's memory in a similar manner. The Patient Health Questionnaire-9 (PHQ-9) asks its user to answer based upon their experience in the previous two weeks \citep{kroenke2001phq}; the Kessler Psychological Distress Scale references the past 30 days, while population health surveys often reference the previous year \citep{mewton2016psychometric, kjellsson2014forgetting}. Clinical assessments rely on client recall for the entirety of their life experience, generally.

Importantly, this condition does not require `accurate' recall. Rather, it presupposes only that the client's current understanding of the summary of their outcome experience remains anchored in the same set of possibilities that would have existed had they measured them. This is unverifiable, but it is reasonable and intuitive to suppose since it is a milder condition than usual. Pretend that we rate our experience of fatigue today as a `2' using the PHQ-9, and that we would only consider a `1' or a `2' as possible per our experience. A week passes and we forget how we answered, or we answer differently in accordance with our current mood. Now we say that we were a `1' at that time. This poses no problem insofar as uncertainty remains anchored upon both `1' and `2.' Anything else can vary. If we think the client or the provider is capable of mentally summarizing experience, then this condition seems safe.

\subsection{Basic Results}

Now, we prove our main results. Recall our setup. We have three variables of interest: $T, Y,$ and $U$. We employ a canonical DAG with line set $\{ T \to Y, U \to T, U \to Y \}$, where $Y$ represents our summary measure of interest. We establish that this alone is sufficient for establishing that $\Omega_{Y|t} = \Omega_{Y|\text{do}(t)}$.

\begin{proposition}
Suppose a DAG with node set $\{ Y, T, U\}$ and line set $\{ T \to Y, U \to T, U \to Y \}$. Then $\Omega_{Y|t} \subseteq \Omega_{Y|\text{do}(t)}$ for all $t \in \Omega_T$. Furthermore, suppose that $f(t|u) > 0$ for all $t \in \Omega_T$ and $u \in \Omega_U$. Then for an arbitrary $t \in \Omega_T$, $\Omega_{Y|t} = \Omega_{Y|\text{do}(t)}$.
\end{proposition}

\begin{proof}
	Suppose the premises. Let both $t, u$ be arbitrary and observe that $f(t|u) > 0$ is equivalent to $f(u|t) > 0$ since both $f(u)$ and $f(t)$ are strictly non-zero on their supports by definition.

	By the law of total probability and the adjustment formula with respect to closing backdoor paths:
	$$f(y | \text{do}(t)) = \sum_{u \in \Omega_U} f(y | t, u) f(u).$$
	However, also observe the following by the same law:
	$$f(y | t) = \sum_{u \in \Omega_U} f(y | t, u) f(u|t).$$

	Observe that both functions---$f(y|t)$ and $f(y | \text{do}(t))$---are functions of $y$ only. Hence, insofar as $f(u | t) > 0$ on the same set of values as $f(u)$, they are \textit{necessarily} positive on the very same set of values.

	To see this, let $t \in \Omega_T$ be arbitrary and assume for a proof by contradiction that there exists a $y_* \in \Omega_{Y|t}$ such that $y_* \notin \Omega_{Y|\text{do}(t)}$. Then $\sum_{u \in \Omega_U} f(y_* | t, u) f(u|t) > 0$ and $\sum_{u \in \Omega_U} f(y_* | t, u) f(u) = 0$. Note that the second conjunct cannot be true since each $f(u) > 0$ for $u \in \Omega_U$ by construction and there must exist a $u_* \in \Omega_U$ such that $f(y_*|t, u_*) > 0$ and $f(u_*|t) > 0$ since $\sum_{u \in \Omega_U} f(y_* | t, u) f(u|t) > 0$.  Consequently, $\Omega_{Y|t} \subseteq \Omega_{Y|\text{do}(t)}$.

	Now, let $t$ again be arbitrary and suppose the opposite direction for contradiction: that there exists a $y_* \in \Omega_{Y|\text{do}(t)}$ such that $y_* \notin \Omega_{Y|t}$. By the previous logic, we know that there must exist a $u_* \in \Omega_U$ such that $f(y_* | t, u_*) > 0$. However, this contradicts our assumption for an arbitrary $t$ that $f(u|t) > 0$ for all $u \in \Omega_U$. This is because $\sum_{u \in \Omega_U} f(y_* | t, u) f(u|t) = 0$ implies that $f(u_*|t)=0$ when it is the case that $f(y_*|t, u_*) > 0$. Therefore, it also follows that $\Omega_{Y|\text{do}(t)} \subseteq \Omega_{Y|t}$. This establishes our statement that $\Omega_{Y|t} = \Omega_{Y|\text{do}(t)}$.
\end{proof}

Proposition~1 is easily extendable to situations with measured random variables $\mathbf{X}$. However, this requires positivity in an expanded form: $f(t|\mathbf{x}, u) > 0$ for all $t, \mathbf{x}, u$. No essential alteration of the logic of the proof is necessary, although the new form of positivity requires additional justification. Proposition~1 is a coherence statement that constrains the modeling outlook of the reasoning agent. Nevertheless, the result still applies in non-Bayesian domains, albeit with a different interpretation, and other schools of causal inference can appropriate it to identify $\Omega$ estimands.

Proposition~1 alone affords us a route to causal inference. A measurement-based care process will always produce at least one point of baseline datum, $y_0$. When $y_0$ is interpreted as uncertain, it becomes $Y_0$ and hence, under only positivity, $\Omega_{Y_0} = \Omega_{Y_0|\text{do}(t=0)}$. The provider and client are thus validly free to reason about $\Omega$ effects that contrast $\Omega_{Y_0|\text{do}(t=0)}$ and $\Omega_{Y|\text{do}(t=1)}$ for some meaningful uncertain function of the treatment data, $Y$. The contrast is rendered weaker inductively because these supports are imperfectly comparable: this is the matching problem, which limits, but does not eliminate, the inductive potential of the contrast.

We can now see how condition III serves us. Since $\Omega_{\hat{Y}} = \Omega_{Y|t=0}$, it follows under positivity that $\Omega_{\hat{Y}} = \Omega_{Y|\text{do}(t=0)}$. This statement closes the comparability gap. If condition III is authentically believable and defensible, inference with respect to treatment effects is rendered much more cogent. Recall that we can also do this because the DAG itself encodes counterfactual information over our set pre- and post-treatment periods. By identifying a single (summary) outcome $Y$ as possibly impacted by $U$ and a binary $T$, we obtain $Y|t=1$ and $Y|t=0$ as possible governors of the entirety of the same time-space. Condition III allows us to use a summary specification for the outcome drawn from a period of the model that exists beyond measurement (since it occurred before the assessment) under the implicit assumption that its outcome space matches that of $Y|t=0$ over the entirety of the period. 

This brings us to condition II. Under positivity and condition II, the expected-average identity affords us the following fact:
$$\text{E}(\hat{Y})-\text{E}(Y|t=1) = \text{Av}\{Y|\text{do}(t=0) \} - \text{Av}\{Y|\text{do}(t=1) \} + (\kappa_{Y|t=1} - \kappa_{\hat{Y}}).$$
In general, enforcing condition II then yields an exact identity since $\kappa_{Y|t=1} - \kappa_{\hat{Y}} = 0$:
$$\text{E}(\hat{Y})-\text{E}(Y|t=1) = \text{Av}\{Y|\text{do}(t=0) \} - \text{Av}\{Y|\text{do}(t=1) \}.$$

Condition II is a reasonable assumption under many circumstances since individuals often reason with this kind of structural bias, as previously argued. However, \textit{enforcing} it can sometimes feel inauthentic to the matters at hand. A common contrast of random variables allows us to arrive at condition II, more or less. Let $\hat{Y}$ represent an estimate of $\bar{Y}|t=0$ over a relevant and defensible period of recall and let $Y=\bar{Y}_{t=1}$. Since $\hat{Y}$ is manifested as a heuristic `averaging' of a ghostly sequence of informally measured experiences in recall, it is likely in many circumstances to be a $\mathbb{U}$ random variable. For $\bar{Y}|t=1$, a heuristic citation of the central limit theorem would suggest approximate $\mathbb{U}$ status, since approximate normality implies approximate symmetry. This is unlikely, however, under thick dependencies. Alternatively, one can argue that an agent overwhelmed by the joint uncertainty of averaging multiple random variables will assign probabilities to the summary function that are at least sum-symmetric, even if non-symmetric.

Bounding a given $\kappa$ also allows for sensitivity analysis. This is perhaps a more cogent route. It can be accomplished with a trivial but useful proposition when the supports in question are sets of integers with no integer missing between minimum and maximum.

\begin{proposition}
Suppose $Y$ is a discrete and regular random variable such that $\text{max}\{\Omega_Y \}=M$ and $|\Omega_Y|=K$. Then $\kappa =  \sum_{i=1}^{K-1} (i-1)f(M-i) - f(M) - 2^{-1}(K-3)$.
\end{proposition}
\begin{proof}
Recall: $\kappa=\text{Av}(Y)-\text{E}(Y)$. Let $m, M$ be the minimum and maximum of the support respectively. For a discrete random variable on integer support $\{m, m+1, m+2, \ldots, M-1, M\}$, the support can be restated as $\{M-K+1, M-K+2, \ldots, M-1, M \}$. Hence, by definition:
\begin{align*}
\kappa & = K^{-1} \sum_{i=0}^{K-1} (M-i) - \sum_{i=0}^{K-1} (M-i)f(M-i) \\
& = \sum_{i=1}^{K-1} i \cdot f(M-i) - 2^{-1}(K-1)
\end{align*}
Now, making use of the fact that $\sum_{i=0}^{K-1} f(M-i) = 1$, it is true that:
$$\sum_{i=1}^{K-1} i \cdot f(M-i) = 1 - f(M) + \sum_{i=1}^{K-1} (i-1)f(M-i).$$
By substitution then:
$$\kappa = \sum_{i=1}^{K-1} (i-1)f(M-i) - f(M) - 2^{-1}(K-3)  .$$
\end{proof}
To see how this proposition is useful, observe $K \in \{2, 3 \}$. Then $\kappa = 2^{-1} - f(M)$ and $f(M-2) - f(M)$ respectively, which are easily bounded with or without additional, reasonable assumptions.

To see how this might apply, specify $Y_0$ to be ``the maximum intensity of symptom experience in recent history'' such that $\Omega_{Y_0} = \{8, 9, 10\}$, but this set is left implicit or is unknown. The client might then choose $y_0=9$ under uncertainty. Then $Y|t=1$ is ``the maximum intensity of symptom experience since treatment commenced,'' empirically instantiated by $\text{max}(y_1, y_2, \ldots, y_m)=y_*$. The client or provider need not consider each random variable atomically, as, insofar as probability is a personal measure with an \textit{intentional focus}, it is sufficient for clinical decision-making to reason about uncertainty from the vantage point of $Y_*$ alone. If the target was then $\text{Av}(Y_0) - \text{Av}(Y_*)$, we could reason about $\kappa_0 - \kappa_* = f_0(M_0-2)-f_*(M_*-2) +f_*(M_*)-f_0(M_0)$ under the assumption that $|\Omega_*|$ also possesses three values. A trivial bound is $|\kappa_0 - \kappa_*| \leq 2$, which can easily be sharpened. When $\text{Av}(Y_0) - \text{Av}(Y_*)$ is intuited as high in numerical magnitude, even the gross bound is useful.

As a final note, one might observe that, once positivity holds, the route that uses condition II is superfluous because $\text{Av}(Y_0) - \text{Av}(Y_*)$ is an $\Omega$ effect that is \textit{already} identified. From this vantage point alone, this is true. Nevertheless, we have supplied it because it gives any difference of expected values causal meaning, in conjunction with a DAG and positivity alone. If one wishes to target a difference in expected values, condition II supplies a road under assumptions markedly milder than those usually claimed.

\section{Data Application}

Data from a case study of a 23-year-old Greek university student referred for counseling is utilized for demonstration. Details are available in the original case report by \cite{tsitsas2014cognitive}. After experiencing severe symptoms of panic and anxiety for six months, the student attended twenty 50-minute sessions of cognitive behavioral therapy spread over five months, completing a Likert ``Everyday Self-Monitoring Scale'' each week for their anxiety and phobia experiences (0: no anxiety; 10: most anxiety without loss of generality). We work with the anxiety scores only. Altogether, $m=19$ scores were provided, in order of measurement: $\{10, 9, 9, 8, 8, 8, 7, 7, 7, 7, 6, 6, 5, 5, 5, 4, 4, 4, 4 \}$. We take creative license with this case study from this point forward.

Since we intend our method to be applicable in the clinic by individuals who do not possess advanced training in probability, we employ a straightforward Bayesian analysis. Say our $\Omega$ estimand of interest is $\text{min}(\Omega_{Y_0 |\text{do}(t=0)}) - \text{min}(\Omega_{Y_*|\text{do}(t=1)})$, where $Y_*$ is the marginal random variable with respect to uncertainty around $y_* = \text{min}(y_1, y_2, \ldots, y_{19}) =4$. We posit that treatment is effective if $\text{min}(\Omega_{Y_0 |\text{do}(t=0)}) > \text{min}(\Omega_{Y_*|\text{do}(t=1)})$. Designate this uncertain event as $\theta$. To fulfill our dictates to be evidence-based, we seek to predicate our judgment of the probability of this event on the observed data, i.e., to derive $\text{Pr}(\theta|y_0, y_*)$. Again, we do not consider the probability of a particular vector $\mathbf{y}$, since a provider and client might not possess intuition for this. Uncertainty in the event $\{y_0, y_*\}$ is likely to be more easily expressed.

We derive $\text{Pr}(\theta| y_0, y_*) = \text{Pr}(y_0, y_*| \theta) \text{Pr}(\theta) / \text{Pr}(y_0, y_*)$ using Bayes' theorem. Say we are incredulous since the 23-year-old possessed severe symptoms at the start and we have a belief that panic is endemic to their nervous system, however terribly founded this belief may be. Hence, we say $\text{Pr}(\theta) =.35$. However, we reason that if $\theta$ is the case, then $\text{Pr}(y_0, y_*| \theta) = .8$; and if $\theta$ \textit{is not the case}, then $\text{Pr}(y_0, y_*| \theta^c) = .3$, where $\theta^c$ is the complementary event. As a consequence, $\text{Pr}(y_0, y_*) = .8\times .35 + .3 \times .65 = .475$ and $\text{Pr}(\theta| y_0, y_*) = .8 \times .35 / .475 = .5895$. Hence, we now possess a weak, albeit positive belief in the efficacy of treatment, based upon the evidence gathered throughout care and this particular choice of metric. The qualitative conclusion is stable to the subjective specifications: varying the prior $\text{Pr}(\theta)$ across $.20$ to $.50$, holding the likelihoods fixed, yields posterior beliefs between $.40$ and $.73$. In every case, we observe a marked update toward efficacy relative to the prior.

Reviewing the logic of our example is beneficial. Besides the client's reception of psychiatric and psycho-therapeutic services, the cited paper makes no mention of life changes during the treatment period that are not a possible function of services. Thus, there is little evidence that the structure of background forces changed. Furthermore, the article reviews a delineated schedule of cognitive-behavioral interventions, which suggests a program designed in disconnection from the client's circumstances. This supports the notion that positivity is satisfied. From the data, $y_0$ and $y_*$ are at hand; positing them as uncertain introduces $Y_0$ and $Y_*$ with supports $\Omega_{Y_0}$ and $\Omega_{Y_* | t=1}$. Since positivity likely holds, these supports are equal to the ones that would have defined the random variables under conditions of intervention, and reasoning about functions of them is an endeavor identified from observed data.

Condition III is also likely satisfied if our judgment that no other life changes occurred throughout the treatment window is cogent. Without it, we would be wary of the strictly non-increasing sequence of scores, which might in some circumstances suggest improvement as a function of cycle or season. If true, this would surface a problem of comparability, although it would not invalidate the basic, self-given contrast of the identified causally-related estimands. We would then exist in the territory of `imperfect matching,' where our evidence suggests an impact in a manner limited by comparison across similar, but non-identical, conditions.

Note that the precision of the number $.5895$ might not be meaningful to the client or provider. However, the calculation allows for the manifestation of coherent, tempered credence. If it is too low (or too high) with respect to their intuition, then more exploration and discussion is due with respect to the working assumptions and specifications. In any case, it provides a measure of belief per the efficacy of treatment that allows for analysis and clinical decision-making. The client can use it alongside their provider to continue on the same course or to alter it, or the provider can use this reasoning process to inform their planning privately. The formal structure also allows for additional clarity for consultations or team meetings.

\subsection{Extension}

What is presented above is a workaday analysis capable of replication in a clinical setting `on the ground' without the assistance of computers. This is intentional, as the person likely to reason causally on a daily basis about treatment efficacy is a non-statistician. However, the proofs of this paper are not restricted to such an enterprise. \textit{Per contra}, they open up causal inference in typical population-based studies. Once a cogent population is granted, positivity is the only further condition required for identification, and insofar as $\Omega$ effects are the target, unmeasured confounding becomes immaterial to identification, although it may still be material for interpretation. This fact is a great boon for program evaluation and research.

\section{Discussion and Limitations}

This paper established that a class of causal estimands---$\Omega$ estimands, defined on supports rather than distributions---is identified in a structural causal model under positivity alone, with no requirement that confounders be known, measured, or adjusted for (Proposition~1). Two further conditions extend the reach of the framework. Support recall (condition III) closes the comparability gap between a sparsely measured baseline and the treatment period and epistemic preservation (condition II) gives ordinary differences of expected values causal meaning through the expected-average identity. Together, these results furnish the provider engaged in measurement-based care with an elementary but principled route to reasoning about treatment efficacy within a single care process.

An obvious objection deserves a direct answer. If $\Omega$ effects are insensitive to confounding because they discard distributional information, are they not also insensitive to most treatment effects? Supports, the objection goes, rarely change. This objection trades on the equivocation of scales and supports flagged in Section~2.1. The co-domain of an instrument is fixed; the support of a random variable, understood as a personal measure of uncertainty, is not. A client for whom scores of 9 and 10 were live possibilities before treatment, and for whom they are no longer live possibilities after, has undergone a change of support that is clinically meaningful. It is the elimination of the possibility of an intensity of suffering, for instance. Precisely this kind of change is what an effective course of treatment purports to accomplish. Nevertheless, the objection carries force within its proper bounds. Treatments that alter only the distribution of probability over an unchanged set of possibilities are invisible to $\Omega$ effects. The framework trades sensitivity for robustness, and the trade should be made knowingly.

Several limitations deserve emphasis. First, positivity is untestable and must be argued case by case via the defense of a DAG. Section~3 supplies heuristics---coarse conceptualizations of treatment, stability of the intuited DAG---but not a test. Second, conditions II and III are likewise unverifiable, although the content of condition III is frequently milder than the recall assumptions already embedded in validated instruments. Third, under our Bayesian interpretation, the results are coherence statements. They constrain what a reasoning agent may believe and do not deliver frequency guarantees. Fourth, the data application is an illustration of reasoning; it is not a validation study. Finally, the expected-average identity and the $\mathbb{U}$ class rest on work published elsewhere \citep{sparkes2024functional}, and the framework would benefit from independent development of both.

Future work should proceed along three lines: intensive and substantive philosophical work that scrutinizes the conditions of this paper; an application to program evaluation, where population definitions free the analysis from the N=1 matching condition; and an empirical study of how often, and for which conditions and treatments, supports are defensibly perceived as shifting in practice. The person-oriented tradition insists that the individual is the proper unit of psychological analysis. We have tried to show that this insistence need not come at the cost of causal rigor.

%%=========================================================================%%
%%  References
%%=========================================================================%%
\bibliographystyle{apacite}
\bibliography{sn-bibliography}

@misc{hernan2010causal,
  title={Causal inference},
  author={Hern{\'a}n, Miguel A and Robins, James M},
  year={2010},
  publisher={CRC Boca Raton, FL}
}

@article{greenland1999causal,
  title={Causal diagrams for epidemiologic research},
  author={Greenland, Sander and Pearl, Judea and Robins, James M},
  journal={Epidemiology},
  volume={10},
  number={1},
  pages={37--48},
  year={1999},
  publisher={LWW},
  doi={10.1097/00001648-199901000-00008}
}

@article{pearl2010foundations,
  title={The foundations of causal inference},
  author={Pearl, Judea},
  journal={Sociological Methodology},
  volume={40},
  number={1},
  pages={75--149},
  year={2010},
  publisher={Wiley Online Library},
  doi={10.1111/j.1467-9531.2010.01228.x}
}

@book{pearl2009causality,
  title={Causality},
  author={Pearl, Judea},
  year={2009},
  publisher={Cambridge university press},
  doi={10.1017/CBO9780511803161}
}

@article{pearl2010introduction,
  title={An introduction to causal inference},
  author={Pearl, Judea},
  journal={The international journal of biostatistics},
  volume={6},
  number={2},
  year={2010},
  publisher={De Gruyter},
  doi={10.2202/1557-4679.1203}
}

@misc{ccbhc2022,
author = {Brown, Jonathan and Breslau, Joshua and Wishon, Allison and Miller, Rachel  and Kase, Courtney and Dunbar, Michael and Stewart, Kate and Briscombe, Brian and Rose, Tyler and Dehus, Eric and DeWitt, Kathryn},
title = {Implementation and Impacts of the Certified Community Behavioral Health Clinic Demonstration: Findings from the National Evaluation},
institution = {Mathematica},
year = {2022},
}

@misc{brown2021implementation,
  title={Implementation and Impacts of the Certified Community Behavioral Health Clinic Demonstration: Findings From the National Evaluation},
  author={Brown, Jonathan and Breslau, Joshua and Wishon, Allison and Miller, Rachel and Kase, Courtney and Dunbar, Michael and Stewart, Kate and Briscombe, Brian and Rose, Tyler and Dehus, Eric and DeWitt, Kathryn},
  institution = {Mathematica},
  year={2021}
}

@misc{mauri2024characterizing,
  title={Characterizing crisis services offered by certified community behavioral health clinics: results from a national survey},
  author={Mauri, Amanda I and Rouhani, Saba and Purtle, Jonathan},
  journal={Psychiatric Services},
  pages={appi--ps},
  year={2024},
  publisher={American Psychiatric Association Washington, DC},
  doi={10.1176/appi.ps.20240152}
}

@article{hu2021brief,
  title={A brief report on certified community behavioral health clinics demonstration program},
  author={Hu, Yuanyuan and Stanhope, Victoria and Matthew, Elizabeth B and Baslock, Daniel M},
  journal={Social Work in Mental Health},
  volume={19},
  number={6},
  pages={534--541},
  year={2021},
  publisher={Taylor \& Francis},
  doi={10.1080/15332985.2021.1929663}
}

@article{uher2017etiology,
  title={Etiology in psychiatry: embracing the reality of poly-gene-environmental causation of mental illness},
  author={Uher, Rudolf and Zwicker, Alyson},
  journal={World Psychiatry},
  volume={16},
  number={2},
  pages={121--129},
  year={2017},
  publisher={Wiley Online Library},
  doi={10.1002/wps.20436}
}

@book{susser2006psychiatric,
  title={Psychiatric epidemiology: searching for the causes of mental disorders},
  author={Susser, Ezra and Schwartz, Sharon and Morabia, Alfredo and Bromet, Evelyn J},
  year={2006},
  publisher={Oxford University Press},
  doi={10.1093/acprof:oso/9780195101812.001.0001}
}

@article{borsboom2017network,
  title={A network theory of mental disorders},
  author={Borsboom, Denny},
  journal={World psychiatry},
  volume={16},
  number={1},
  pages={5--13},
  year={2017},
  publisher={Wiley Online Library},
  doi={10.1002/wps.20375}
}

@article{sparkes2024functional,
  title={The functional average treatment effect},
  author={Sparkes, Shane and Garcia, Erika and Zhang, Lu},
  journal={Journal of Causal Inference},
  volume={12},
  number={1},
  pages={20230076},
  year={2024},
  publisher={De Gruyter},
  doi={10.1515/jci-2023-0076}
}

@article{tsitsas2014cognitive,
  title={A cognitive-behavior therapy applied to a social anxiety disorder and a specific phobia, case study},
  author={Tsitsas, George D and Paschali, Antonia A},
  journal={Health psychology research},
  volume={2},
  number={3},
  year={2014},
  publisher={Open Medical Publishing},
  doi={10.4081/hpr.2014.1603}
}

@book{deFinetti1974Theory,
	title     = {Theory of Probability: A Critical Introductory Treatment},
	author    = {de Finetti, Bruno},
	volume    = {1},
	year      = {1974},
	publisher = {John Wiley \& Sons},
	address   = {New York},
	note      = {Translated by Antonio Mach{\`i} and Adrian Smith},
	doi       = {10.1002/9781119286387}
}

@article{Holland1986,
	author    = {Paul W. Holland},
	title     = {Statistics and Causal Inference},
	journal   = {Journal of the American Statistical Association},
	volume    = {81},
	number    = {396},
	pages     = {945--960},
	year      = {1986},
	publisher = {Taylor \& Francis},
	doi       = {10.1080/01621459.1986.10478354}
}

@article{backman1999case,
	title={Case studies, single-subject research, and n of 1 randomized trials: Comparisons and Contrasts: 1},
	author={Backman, Catherine L and Harris, Susan R},
	journal={American journal of physical medicine \& rehabilitation},
	volume={78},
	number={2},
	pages={170--176},
	year={1999},
	publisher={LWW},
	doi={10.1097/00002060-199903000-00022}
}

@article{kwasnicka2019challenges,
	title={Challenges and solutions for N-of-1 design studies in health psychology},
	author={Kwasnicka, Dominika and Inauen, Jennifer and Nieuwenboom, Wim and Nurmi, Johanna and Schneider, Annegret and Short, Camille E and Dekkers, Tessa and Williams, A Jess and Bierbauer, Walter and Haukkala, Ari and others},
	journal={Health Psychology Review},
	volume={13},
	number={2},
	pages={163--178},
	year={2019},
	publisher={Taylor \& Francis},
	doi={10.1080/17437199.2018.1564627}
}

@article{fountzilas2022clinical,
	title={Clinical trial design in the era of precision medicine},
	author={Fountzilas, Elena and Tsimberidou, Apostolia M and Vo, Henry Hiep and Kurzrock, Razelle},
	journal={Genome medicine},
	volume={14},
	number={1},
	pages={101},
	year={2022},
	publisher={Springer},
	doi={10.1186/s13073-022-01102-1}
}

@book{kazdin2020single,
	title={Single-case research designs: Methods for clinical and applied settings},
	author={Kazdin, Alan E},
	year={2020},
	publisher={Oxford University Press}
}

@article{lundervold2000best,
	title={The best kept secret in counseling: Single-case (N= 1) experimental designs},
	author={Lundervold, Duane A and Belwood, Marilyn F},
	journal={Journal of Counseling \& Development},
	volume={78},
	number={1},
	pages={92--102},
	year={2000},
	publisher={Wiley Online Library},
	doi={10.1002/j.1556-6676.2000.tb02565.x}
}

@article{kratochwill2013single,
	title={Single-case intervention research design standards},
	author={Kratochwill, Thomas R and Hitchcock, John H and Horner, Robert H and Levin, Joel R and Odom, Samuel L and Rindskopf, David M and Shadish, William R},
	journal={Remedial and Special Education},
	volume={34},
	number={1},
	pages={26--38},
	year={2013},
	publisher={Sage Publications Sage CA: Los Angeles, CA},
	doi={10.1177/0741932512452794}
}

@article{kratochwill2015single,
	title={Single-case research design and analysis: An overview},
	author={Kratochwill, Thomas R},
	journal={single-case Research Design and Analysis (psychology Revivals)},
	pages={1--14},
	year={2015},
	publisher={Routledge},
	doi={10.4324/9781315725994}
}

@article{duan2022personalized,
	title={Personalized data science and personalized (N-of-1) trials: Promising paradigms for individualized health care},
	author={Duan, Naihua and Norman, Daniel and Schmid, Christopher and Sim, Ida and Kravitz, Richard L},
	journal={Harvard data science review},
	volume={4},
	number={SI3},
	pages={10--1162},
	year={2022},
	doi={10.1162/99608f92.8439a336}
}

@article{kravitz2022conduct,
	title={Conduct and implementation of personalized trials in research and practice},
	author={Kravitz, Richard L and Duan, Naihua},
	journal={Harvard data science review},
	volume={4},
	number={SI3},
	pages={10--1162},
	year={2022},
	doi={10.1162/99608f92.901255e7}
}

@article{duan2013single,
	title={Single-patient (n-of-1) trials: a pragmatic clinical decision methodology for patient-centered comparative effectiveness research},
	author={Duan, Naihua and Kravitz, Richard L and Schmid, Christopher H},
	journal={Journal of clinical epidemiology},
	volume={66},
	number={8},
	pages={S21--S28},
	year={2013},
	publisher={Elsevier},
	doi={10.1016/j.jclinepi.2013.04.006}
}

@article{molenaar2009new,
	title={The new person-specific paradigm in psychology},
	author={Molenaar, Peter CM and Campbell, Cynthia G},
	journal={Current directions in psychological science},
	volume={18},
	number={2},
	pages={112--117},
	year={2009},
	publisher={SAGE Publications Sage CA: Los Angeles, CA},
	doi={10.1111/j.1467-8721.2009.01619.x}
}

@article{molenaar2004manifesto,
	title={A manifesto on psychology as idiographic science: Bringing the person back into scientific psychology, this time forever},
	author={Molenaar, Peter CM},
	journal={Measurement},
	volume={2},
	number={4},
	pages={201--218},
	year={2004},
	publisher={Taylor \& Francis},
	doi={10.1207/s15366359mea0204_1}
}

@article{molenaar2015relation,
	title={On the relation between person-oriented and subject-specific approaches},
	author={Molenaar, Peter CM},
	journal={Journal for Person-Oriented Research},
	volume={1},
	number={1-2},
	pages={34--41},
	year={2015},
	doi={10.17505/jpor.2015.04}
}

@article{de2026n,
	title={N-of-1 trials in clinical research: Methodological foundations, statistical approaches and implementation challenges},
	author={De Carvalho, Marcos Clint Leal and de Matos Dourado Sim{\~o}es, Matheus and Adler, Ubiratan Cardinalli and Junior, Antonio Brazil Viana and de Sousa, Caren N{\'a}dia Soares and Sanders, Lia Lira Olivier},
	journal={British Journal of Clinical Pharmacology},
	volume={92},
	number={3},
	pages={809--821},
	year={2026},
	publisher={Wiley Online Library},
	doi={10.1002/bcp.70382}
}

@book{austin1962,
	title     = {How to Do Things with Words},
	author    = {Austin, J. L.},
	year      = {1962},
	publisher = {Oxford University Press},
	address   = {Oxford, UK}
}

@book{searle1979expression,
	title     = {Expression and Meaning: Studies in the Theory of Speech Acts},
	author    = {Searle, John R.},
	year      = {1979},
	publisher = {Cambridge University Press},
	address   = {Cambridge, England}
}

@article{kroenke2001phq,
	title={The PHQ-9: validity of a brief depression severity measure},
	author={Kroenke, Kurt and Spitzer, Robert L and Williams, Janet BW},
	journal={Journal of general internal medicine},
	volume={16},
	number={9},
	pages={606--613},
	year={2001},
	publisher={Wiley Online Library},
	doi={10.1046/j.1525-1497.2001.016009606.x},
}

@article{frank2023toward,
  author  = {Frank, Richard G. and Shim, Ruth S.},
  title   = {Toward greater accountability in mental health care},
  journal = {Psychiatric Services},
  year    = {2023},
  volume  = {74},
  number  = {2},
  pages   = {182--187},
  doi     = {10.1176/appi.ps.20220097},
}

@article{kjellsson2014forgetting,
  author  = {Kjellsson, Gustav and Clarke, Philip and Gerdtham, Ulf-G.},
  title   = {Forgetting to remember or remembering to forget: {A} study of the recall period length in health care survey questions},
  journal = {Journal of Health Economics},
  year    = {2014},
  volume  = {35},
  pages   = {34--46},
  doi     = {10.1016/j.jhealeco.2014.01.007},
}

@article{mewton2016psychometric,
  author  = {Mewton, Louise and Kessler, Ronald C. and Slade, Tim and Hobbs, Megan J. and Brownhill, Louise and Birrell, Louise and Tonks, Zoe and Teesson, Maree and Newton, Nicola and Chapman, Cath and others},
  title   = {The psychometric properties of the {Kessler} {Psychological} {Distress} {Scale} ({K6}) in a general population sample of adolescents},
  journal = {Psychological Assessment},
  year    = {2016},
  volume  = {28},
  number  = {10},
  pages   = {1232--1242},
  doi     = {10.1037/pas0000239},
}

\end{document}